\documentclass{cccg23}
\usepackage{graphicx,amssymb,amsmath}

\usepackage{xspace}
\usepackage{hyperref}
\usepackage{cleveref}
\usepackage{enumitem}
\usepackage{thm-restate}
\usepackage{caption}
\usepackage{subcaption}
\newcommand{\R}{\ensuremath{\mathbb{R}}}
\newcommand{\poly}{\ensuremath{\mathsf{P}}\xspace}
\newcommand{\NP}{\ensuremath{\mathsf{NP}}\xspace}
\newcommand{\problemname}[1]{\ensuremath{#1}\text{\nobreakdash-\texttt{MinNN}}}
\DeclareMathOperator{\rel}{rel}
\DeclareMathOperator{\nn}{nn}
\newtheorem{definition}{Definition}
\crefname{obs}{observation}{observations}

\title{%The Complexity of
Reducing Nearest Neighbor Training Sets Optimally and Exactly
}

\author{Josiah Rohrer\thanks{\mbox{Department of Mathematics,
        ETH Zurich,} {\tt rohrerj@student.ethz.ch}}
        \and
        Simon Weber\thanks{\mbox{Department of Computer Science,  ETH Zurich,} {\tt  simon.weber@inf.ethz.ch}\newline Simon Weber is supported by the Swiss National Science Foundation under project no. 204320.}}

\index{Rohrer, Josiah}
\index{Weber, Simon}

\begin{document}
\thispagestyle{empty}
\maketitle

\begin{abstract}
In nearest-neighbor classification, a \emph{training set} $P$ of points in $\R^d$ with given classification is used to classify every point in $\R^d$: Every point gets the same classification as its nearest neighbor in $P$. Recently, Eppstein [SOSA'22] developed an algorithm to detect the \emph{relevant} training points, those points $p\in P$, such that $P$ and $P\setminus\{p\}$ induce different classifications. We investigate the problem of finding the \emph{minimum cardinality reduced training set} $P'\subseteq P$ such that $P$ and $P'$ induce the same classification. We show that the set of relevant points is such a minimum cardinality reduced training set if $P$ is in general position. Furthermore, we show that finding a minimum cardinality reduced training set for possibly degenerate $P$ is in \poly for $d=1$, and \NP-complete for $d\geq 2$.
\end{abstract}

\section{Introduction}
While it is one of the oldest and simplest to describe classification techniques, \emph{nearest-neighbor classification}~\cite{cover1967nearestneighbor} is still a widely-used method in supervised learning. A \emph{training set} $P$ consisting of data points in $\R^d$ labelled with their known classifications is used to classify new points in $\R^d\setminus P$. A point $q$ gets the same classification as its nearest neighbor in $P$ (ties are either broken by some fixed rule or the point gets multiple classifications).

There are many variations of nearest-neighbor classification, such as $k$-nearest neighbor~\cite{cover1967nearestneighbor}, where a point gets the majority classification among its $k$ nearest neighbors, and approximate versions of nearest neighbor~\cite{liu2004approximate}.  In this paper we only consider the basic version described above.

Nearest neighbor classification and the need to implement it efficiently has motivated many concepts in computational geometry. Voronoi diagrams describe the decomposition of $\R^d$ into cells with the same nearest neighbor, and thus the same nearest neighbor classification~\cite{aurenhammer2013voronoi}. They have been extended to higher-order Voronoi diagrams~\cite{aurenhammer2013voronoi}, which analogously describe the cells with the same $k$ nearest neighbors. Much research has gone into efficiently computing Voronoi diagrams~\cite{dwyer1991voronoi,guibas1992incremental,watson1981voronoi} as well as point-location techniques to locate the cell of a Voronoi diagram containing a given query point~\cite{guibas1992incremental,preparata1992pointlocation}. Any technique based on explicitly storing or computing the Voronoi diagram of the training set is infeasible for higher-dimensional data, since the complexity of the Voronoi diagram of $n$ points in dimension $d$ can reach $\Theta(n^{\lceil d/2\rceil})$~\cite{seidel1990exactVoronoi}. For moderate and high dimensions, various methods for approximate nearest neighbor searching have been developed, such as quadtree-based data structures~\cite{arya1998quadtree1,arya2009quadtree2,chan1998approximate,eppstein2008skip} and locality-sensitive hashing~\cite{andoni2008hashing,datar2004hashing,gionis1999similarity,indyk1998hashing,liu2004approximate}. These methods avoid the exponential dependency on $d$, but are still only marginally better than naively computing the nearest neighbor of a query point by searching through the complete training set.

Instead of improving nearest neighbor algorithms and data structures, significant time and storage can be saved by reducing the size of the training set. A common approach to reducing the training set in a \emph{lossless} manner (without changing the classification of any query point) is to remove all non-relevant points. A \emph{relevant point} (sometimes also called \emph{border point}) is a point whose individual omission changes the classification~\cite{clarkson1994outputsensitive}. One can show that removing all non-relevant points at once yields a training set inducing the same classification as the original training set. A series of algorithms have been developed to efficiently compute the set of relevant points. The current best algorithm due to Flores-Velazco~\cite{floresvelazco2022improvedrelevant} finds the set of relevant points in any fixed dimension in $O(nk^2)$, where $k$ is the number of relevant points. This algorithm is a slightly adjusted version of the algorithm of Eppstein~\cite{eppstein2022relevantpoints}.

\begin{figure}
    \centering
    \includegraphics[width=0.8\columnwidth, keepaspectratio]{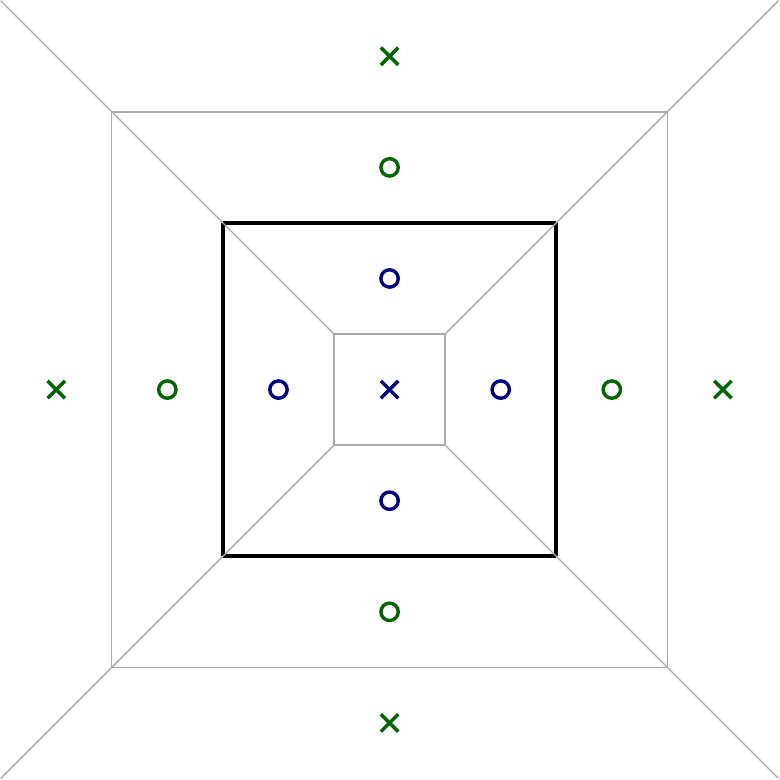}
    \caption{A training set for which the relevant points (circles) do not contain the unique minimum cardinality reduced training set (crosses).}
    \label{fig:relevantPointsSuck}
\end{figure}

The set of relevant points is not necessarily the smallest subset of the training points inducing the same classification. In fact, it is not even guaranteed that the set of relevant points contains such a smallest subset. This is illustrated in \Cref{fig:relevantPointsSuck}. In the paper introducing his algorithm to find the relevant points, Eppstein~\cite{eppstein2022relevantpoints} conjectures that in high dimensions, finding such a smallest subset with the same classification is a much harder problem than finding the relevant points. In this paper, we show that high dimensions are not needed, and this problem is already \NP-hard for binary classification in dimensions $d\geq 2$.

Note that lossless training set reduction is not the only studied method. We discuss alternative methods which are allowed to (slightly) change the induced classification later in \Cref{sec:relatedwork}.

\subsection{Definitions}
\begin{definition}
    A \emph{labelled point set} $(m, P, c)$ is given by an integer $m$, a set $P\subset \R^d$ of size $n$, and a classification function $c:P\rightarrow [m]$. We call $c(p)$ the \emph{label} of $p$.
\end{definition}

By $d(p,q)$ we denote the euclidean distance between two points $p,q$. We write $\nn(q,P)$ for the set of nearest neighbors of $q$ in $P$. We say a point set is in \emph{general position} if it contains no three collinear points and no four cocircular points.

\begin{definition}
    A labelled point set $(m, P, c)$ induces the nearest neighbor classification $f:\R^d\rightarrow 2^{[m]}$, where
    \[ f(q)=\{ c(p) \;|\; p\in nn(q,P)\} .\]
\end{definition}

\begin{definition}
    A \emph{reduced training set} of some labelled point set $(m, P, c)$ is a set $Q\subseteq P$ such that $(m, Q, c\vert_Q)$ induces the same nearest neighbor classification as $(m, P, c)$.
\end{definition}

\begin{definition}
The decision problem \problemname{d} is to decide if there exists a reduced training set of a given labelled point set $(m,P,c)$ of at most $k$ points.
\end{definition}

\begin{definition}
    A point $p\in P$ is a \emph{relevant point} if $(m,P\setminus\{p\},c\vert_{P\setminus\{p\}})$ and $(m,P,c)$ induce different nearest neighbor classifications. We write $\rel(P)$ for the set of all relevant points.
\end{definition}

\subsection{Results}
We are now ready to state our results. As our first result, we show that in the case of training sets in general position, \problemname{d} can be solved easily, since the relevant points already form a solution.

\begin{restatable}{theorem}{generalposition}\label{thm:generalposition}
For an instance of \problemname{d} where $P$ is in general position, $\rel(P)$ is the unique solution.
\end{restatable}

If this assumption of general position is not given, the set of relevant points is not guaranteed to be an optimal solution. We show that generally, finding a minimum cardinality reduced training set is only feasible in dimension one, and \NP-complete otherwise.

\begin{theorem}\label{thm:dim1}
    \problemname{1} is in \poly.
\end{theorem}
\begin{restatable}{theorem}{NPhardness}\label{thm:NP}
    For any fixed dimension $d\geq 2$, \problemname{d} is \NP-complete, even for binary classification, i.e., $m=2$.
\end{restatable}

\subsection{Discussion}
When data points are independently sampled from a probability distribution with (e.g., Gaussian) noise, the resulting data set is in general position with probability~$1$. \Cref{thm:generalposition} thus implies that in practice, when a classification model is trained from high-precision data coming from a noisy source, computing the relevant points using the algorithms of Eppstein~\cite{eppstein2022relevantpoints} or Flores\nobreakdash-Velazco~\cite{floresvelazco2022improvedrelevant} is an efficient way to reduce the size of the training set to the optimum in a lossless fashion. The only way to reduce the size of the training set any further is to accept some small errors. While our results do not imply hardness of approximative training set reduction, \Cref{thm:NP} shows that it cannot be achieved efficiently by first ``de-noising'' the training set, and then finding the minimum cardinality reduced training set.

\subsection{Related Work}\label{sec:relatedwork}
Much of the work on nearest neighbor training set reduction has focused on detecting relevant points. Since the number of relevant points $k$ is expected to be very small compared to the total number of points $n$, algorithms to find relevant points are ideally output-sensitive. The first such algorithm has been found by Clarkson in 1994~\cite{clarkson1994outputsensitive}. Bremner et al.~\cite{bremner2005outputsensitive} improved on Clarkson's algorithm for two-dimensional data with two labels. Recently, Eppstein~\cite{eppstein2022relevantpoints} gave an algorithm for all dimensions and any number of labels, which is based on the simple geometric primitives of computing Euclidean minimum spanning trees and extreme points of point sets. Flores-Velazco~\cite{floresvelazco2022improvedrelevant} then showed that the Euclidean minimum spanning tree step can be skipped, yielding an $O(nk^2)$ algorithm for any constant dimension $d$.

Lossy reduction of nearest neighbor training sets, i.e., reduction in a way that slightly changes the classification, is often called \emph{nearest neighbor condensation} in the literature. The most common concept in condensation is that of \emph{consistent subsets}, introduced by Hart in 1968~\cite{hart1968consistent}. A consistent subset is a subset of the training points that induces the same classification on the original training set, but not necessarily on all points of $\R^d$. It is known that computing a minimum cardinality consistent subset is \NP-complete, for any number of labels $m\geq 2$~\cite{khodamoradi2018twolabels,wilfong1991consistent}. \emph{Selective subsets}~\cite{ritter1975selective} are subsets fulfilling a stronger condition than consistent subsets. Here, the distance from every point $p$ in the original training set to a point with the same classification in the subset must be smaller than the distance from $p$ to the nearest point of different classification in the original training set. Minimum cardinality selective subsets are also \NP-complete to compute~\cite{zukhba2010selective}. Flores-Velazco and Mount~\cite{floresvelazco2020coreset} introduced the approximative notions of $\alpha$-consistency and $\alpha$-selectivity and showed that it is \NP-hard not only to find minimum cardinality $\alpha$-consistent and $\alpha$-selective subsets, but also to approximate their size beyond certain approximation factors. Due to all of these \NP-hardness results, much of the recent research has focused on heuristic methods providing some guarantee on the resulting subset size~\cite{jankowski2004survey}.

So far we have only discussed training set reduction by taking a subset of the original data. Of course, another option is to construct a completely new training set that (approximately) induces the same classification as the original data while containing fewer points. Heath and Kasif~\cite{heath1993minimalvoronoi} showed that an exact version of this approach is hopeless, since they show that finding the minimum number of points needed to create a Voronoi diagram containing a given polygonal tesselation as a substructure is \NP-hard. This is a partial explanation to why this approach has not been studied much by the nearest neighbor community.

\subsection{Proof Techniques}
The proof of \Cref{thm:generalposition} is very straightforward. It makes use of the observation that in every reduced training set $Q\subseteq P$, every Voronoi wall of $P$ between two regions of different classifications must lie in the bisecting hyperplane of some pair of points. If $P$ is in general position, no two pairs of points have the same bisecting hyperplane.

To prove \Cref{thm:dim1} we provide a reduction from \problemname{1} to the problem of finding a maximum weight independent set on interval graphs, which is solvable in polynomial time~\cite{hsiao1992maxweightindependent}.

\begin{figure}
    \centering
    \includegraphics[width=\columnwidth, keepaspectratio]{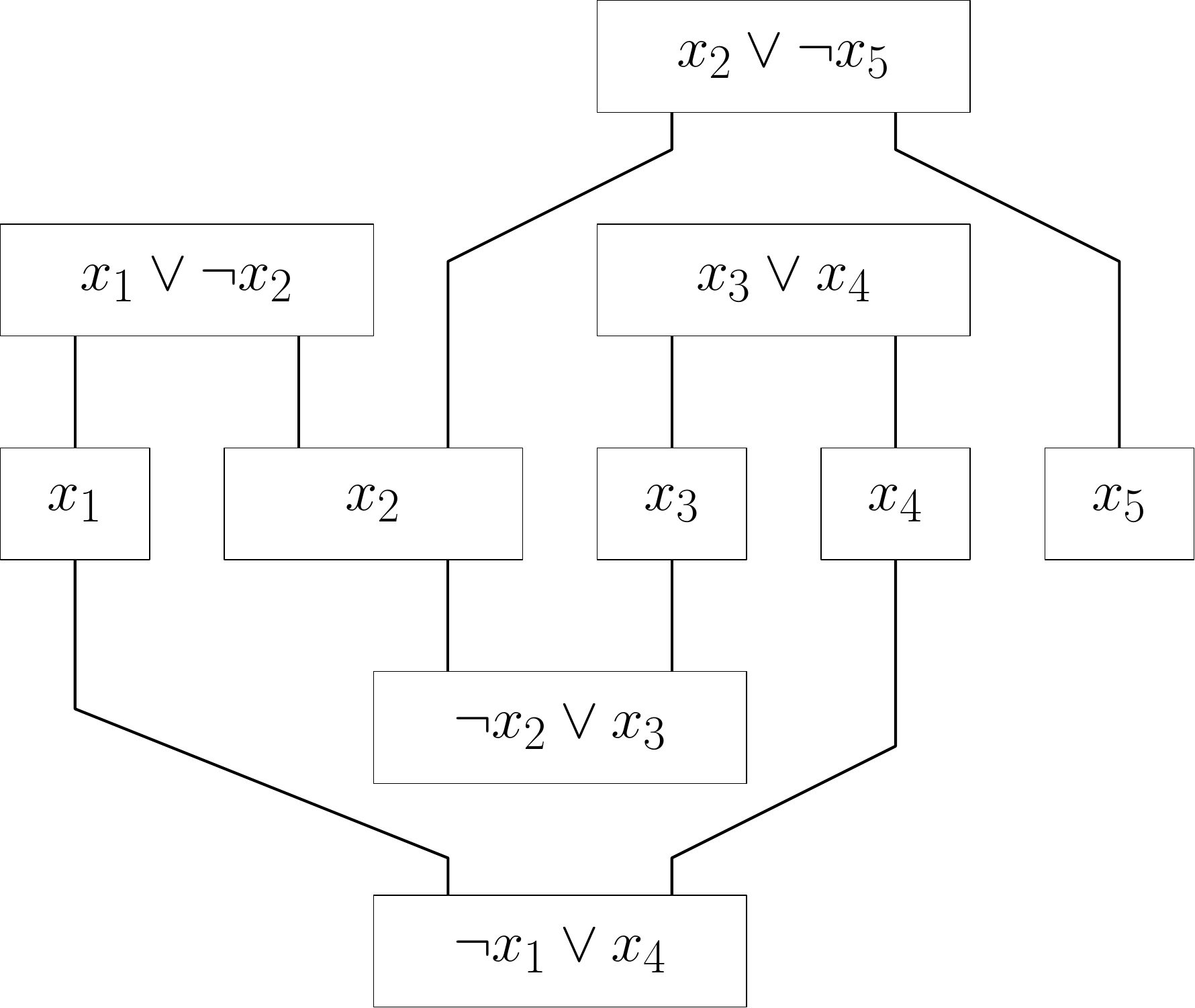}
    \caption{The type of embedding of the variable-clause graph used in the proof of \Cref{thm:NP}.}
    \label{fig:graphembedding}
\end{figure}

Our proof of \NP-hardness for \Cref{thm:NP} is similar to the proof of Heath and Kasif for the \NP-hardness of the problem of finding Voronoi covers, which works by reduction from planar 3SAT. The proofs have two major differences. On one hand, a solution to \problemname{d} must have the same classification on all of $\R^d$. In contrast, in the Voronoi cover problem, only a fixed tesselation needs to appear as a substructure in the Voronoi diagram. This means that we have to be more careful about introducing additional Voronoi walls in our training set. On the other hand, any solution to \problemname{d} must be a subset of the training set, while in the Voronoi cover problem, arbitrary points are allowed. This gives us more control about the structure of possible solutions, and allows us to exclude unwanted solutions more easily.

Our proof works by reduction from the problem \emph{V\nobreakdash-cycle max2SAT}, a variant of max2SAT in which the bipartite variable-clause graph remains planar even after adding a Hamiltonian cycle through the vertices $x_1,\ldots,x_n$ corresponding to the variables. This problem has been proven \NP-hard recently by Buchin et al.~\cite{buchin2020planarmax2Satwithcycle}. The planarity of this graph guarantees that we can efficiently find an embedding of the graph of a certain type, as shown in \Cref{fig:graphembedding}. Then, every box corresponding to a variable $x_i$ is replaced by a \emph{variable gadget}, a labelled point set with two strict subsets with the same induced classification. The choice between these two subsets indicates the value of the variable $x_i$. This value is then passed along the edges of the graph by \emph{channels}. Finally, each box corresponding to a clause $C_j$ is replaced by a \emph{clause gadget}, a labelled point set for which the size of a minimum cardinality reduced training set is decreased by one if and only if at least one of two other points is already present. These two points belong to the channels feeding in the values of the two literals of the clause. The clause gadget thus requires one fewer point iff the clause is fulfilled. The size of the minimum cardinality reduced training set for the resulting labelled point set thus allows us to determine the largest number of simultaneously fulfillable clauses in the V-cycle max2SAT instance.

The main technical challenges in this reduction are \begin{enumerate}[label=(\roman*)]
    \item avoiding unwanted interaction between gadgets, since points with different labels can interact over large distances in empty space, and
    \item ensuring that the reduction yields a point set of only a polynomial number of points, with polynomial-sized coordinates.
\end{enumerate}

\subsection{Paper Overview}
We prove \Cref{thm:generalposition} in \Cref{sec:generalposition}. Then, we prove \Cref{thm:dim1} in \Cref{sec:1d}. Finally, \Cref{sec:NP} is dedicated to proving \Cref{thm:NP}.
%, but due to space constraints we defer the proofs of most technical parts to the appendix or the full version of the paper.

%%%%%

\section{General Position}\label{sec:generalposition}
In this section we wish to prove \Cref{thm:generalposition}:
\generalposition*
Since we already know that for a labelled point set $(m,P,c)$ the set $\rel(P)$ of relevant points is a reduced training set, it suffices to show that any reduced training set must include all relevant points.

Let us first introduce a few definitions. A \emph{Voronoi wall} is a cell of the Voronoi diagram of dimension \mbox{$d-1$}. These walls separate two fully-dimensional Voronoi cells. A Voronoi wall that separates two cells in which the nearest neighbor rule induced by $P$ gives a different classification is said to be part of the \emph{decision boundary}.

\begin{obs}\label{obs:boundarywalls}
For any reduced training set given by $Q\subseteq P$, a Voronoi wall of $P$ which is part of the decision boundary must be (a subset of) some Voronoi wall of $Q$ as well, since $Q$ must induce the same nearest neighbor classification as $P$.
\end{obs}
 Note that every Voronoi wall $W$ is a subset of the bisecting hyperplane of the two points belonging to the incident Voronoi cells. Since $W$ is $d-1$-dimensional, it uniquely determines this hyperplane. We next show that under the general position assumption, every hyperplane can be the bisecting hyperplane of at most one pair of points.

\begin{lemma}\label[lemma]{lem:twobisectors}
For any point set $Q\subset \R^d$ in general position, no two distinct pairs of points in $Q$ have the same bisecting hyperplane.
\end{lemma}
\begin{proof}
Towards a contradiction, assume $a,b\in Q$ and $c,d\in Q$ have the same bisecting hyperplane. First note that $a,b,c,d$ must all be distinct, since $a,b$ and $a,d$ have different bisecting hyperplanes if $b\neq d$.

The points $a,b,c,d$ can not be collinear, since general position assumption requires that no three points are collinear. Since $a,b$ and $c,d$ have the same bisecting hyperplane, the lines $ab$ and $cd$ must be parallel. Thus, $a,b,c,d$ must lie on a common plane. Furthermore, they must be the corners of an isoceles trapezoid, a \emph{cyclic quadrilateral}~\cite{usiskin2008quadrilaterals}. Thus, $a,b,c,d$ are cocircular, forming a contradiction with the general position assumption.
\end{proof}

We know by the proofs in \cite{eppstein2022relevantpoints} that every relevant point $p\in \rel(P)$ shares a Voronoi wall with some point $q$ with a label $c(q)\neq c(p)$, i.e., a wall that is part of the decision boundary. By \Cref{obs:boundarywalls} and \Cref{lem:twobisectors}, we know that $p,q$ must therefore be part of every reduced training set. Since this holds for every relevant point $p$, any reduced training set given by $Q$ must contain the set $\rel(P)$, and \Cref{thm:generalposition} follows.

%%%%%

\section{The One-Dimensional Case}\label{sec:1d}
In this section, we provide a polynomial-time algorithm to find a minimum cardinality reduced training set in $\R^1$. In other words, we prove \Cref{thm:dim1}, $\problemname{1}\in\poly$.

The classification induced by the given training set $(m,P,c)$ decomposes $\R^1$ into a set $C=\{C_1,\ldots,C_t\}$ of open intervals of equal classification. We call the set $B=\{b_1,\ldots,b_{t-1}\}$ of points between these open intervals the \emph{decision boundary points}. We begin with some observations holding for any reduced training set $(m,Q,c\vert_Q)$. First, for any $i\in[t]$, $Q\cap C_i\neq\emptyset$. Second, for any $i\in[t-1]$, $b_i$ is the midpoint between its closest larger and smaller neighbor in $Q$. Finally, for any minimum cardinality reduced training set, we must have $Q\cap C_i\leq 2$.

Intuitively, towards a small reduced training set, we have to find a subset of $P$ which often contains only one point per interval $C_i$, with this point being involved in defining both $b_{i-1}$ and $b_i$. We formalize this in the following notion of a \emph{chain}, as illustrated in \Cref{fig:kchain}.

\begin{definition}
A $k$-chain is a set $Q:=\{q_1,\ldots, q_k\}\subseteq P$, for which there exists an integer $i$ such that:\\
(i) for any $j\in\{1,\ldots,k\}$, $q_j\in C_{i+j}$, and\\
(ii) for any $j\in\{1,\ldots,k-1\}$, $\frac{q_j+q_{j+1}}{2}=b_{i+j}$. We then say that $P'$ \emph{covers} the boundary points $b_i,\ldots,b_{i+k-1}$.
\end{definition}

\begin{figure}
    \centering
    \includegraphics[width=\columnwidth,keepaspectratio]{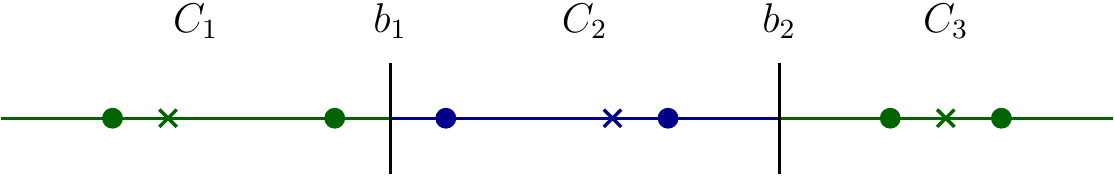}
    \caption{A labelled point set in $\R^1$ and a $3$-chain (crosses) covering $b_1$ and $b_2$.}
    \label{fig:kchain}
\end{figure}

We say that two chains $Q,Q'$ are \emph{compatible}, if the intervals $[\min(Q),\max(Q)]$ and $[\min(Q'),\max(Q')]$ are disjoint. Compatible chains therefore cover disjoint sets of boundary points.

We now see that any minimum cardinality reduced training set must be the union of pairwise compatible chains. Any set of pairwise compatible chains can furthermore be completed to a reduced training set by adding $2$-chains (consisting of relevant points) to cover the remaining uncovered boundary points.

If our reduced training set is a union of $k_1-,\ldots,k_\ell-$chains, the total number of points is $2(t-1)-\sum_{i=1}^\ell (k_i-2)$, since every $k_i$-chain allows us to save a point in the $k_i-2$ intervals between their first and last covered boundary points, compared to a naive solution with $2(t-1)$ points consisting only of $2$-chains.

We are now ready to state our complete algorithm. First, we compute the set of boundary points and the set of all chains. Note that we can easily compute the set of all chains in $O(n^2)$, and that there are at most $n^2$ of them. Then, we associate each chain $Q$ with the interval $[\min(Q),\max(Q)]$. Note now that a set of pairwise compatible chains is an independent set in the interval graph given by these intervals. We give each $2$-chain a tiny weight $\epsilon>0$, and each $k$-chain for $k>2$ the weight $k-2$. Finally, we use the dynamic programming approach of \cite{hsiao1992maxweightindependent} to find the \emph{maximum weight independent set (MWIS)} within this graph. This algorithm is linear in the number of vertices, thus takes $O(n^2)$ in our case.

The resulting independent set corresponds to an inclusion-maximal independent set, with the maximum weight among all such sets. Its corresponding chains thus cover all boundary points, and their union is a minimum cardinality reduced training set.

%%%%%
\section{NP-Completeness}\label{sec:NP}

In this section, we prove \Cref{thm:NP}.
\NPhardness*

We run the proof of \NP-hardness for $d=2$, since any instance of \problemname{2} can be embedded in a $2$-dimensional subspace of $\R^d$ to yield an instance of \problemname{d}.
The proof works by reduction from the following problem, which has been proven \NP-hard by Buchin et al. \cite{buchin2020planarmax2Satwithcycle}\footnote{A proof can be found in the appendix of the arXiv preprint~\cite{buchin2019planarmax2SatwithcyclearXiv}.}.

\begin{definition}\label[definition]{def:vcycle2SAT}
    A conjunctive normal form formula $\phi=C_1\wedge \ldots\wedge C_b$ over the variables $x_0,\ldots,x_{a-1}$ and an integer $k$ form an instance of the \emph{V-cycle max2SAT} problem, if every clause $C_i$ consists of at most two literals and the graph \mbox{$G_\phi=(V,E)$} is planar, where
\begin{align*}
V&=\{C_i\;|\;i\in [b]\}\cup \{x_i\;|\;0\leq i\leq a-1\},\\
E&=\{(C_i,x_j)\;|\;x_j\in C_i\}\cup \{(x_i,x_{i+1\;\mathrm{mod}\;a})\;|\;i\in [a]\}.
\end{align*}
The task is to decide whether there exists an assignment of the variables $x_0,\ldots,x_{a-1}$, such that at least $k$ clauses of $\phi$ are fulfilled.
\end{definition}

Note that the graph $G'_\phi$ obtained by replacing every clause vertex $C_i$ (of degree $2$) in $G_\phi$ by an edge is Hamiltonian, with the Hamiltonian cycle $(x_0,\ldots, x_{a-1},x_0)$.

An $n$-book embedding~\cite{bernhart1979bookthickness} of a graph is an embedding into the space of $n$ half-planes with the same bounding line $\ell$, such that all vertices are distinct points on $\ell$, every edge intersects the interior of exactly one half-plane, and no edges intersect except at common endpoints. The book-thickness $bt(G)$ of a graph is the minimum $n$ such that $G$ has an $n$-book embedding.
\begin{lemma}[\cite{bernhart1979bookthickness}]
    For every planar (sub-)Hamiltonian graph $G$, $bt(G)\leq 2$.
\end{lemma}
The proof of this lemma is constructive, and if a Hamiltonian path in $G$ is known, the construction can be performed in polynomial time. Note now that a $2$-book embedding of $G'_\phi$ can be turned into a planar embedding of $\phi$ of the form in \Cref{fig:graphembedding} by reintroducing the clause vertices and removing the Hamiltonian path.

Given this embedding of $\phi$, we will now construct a labelled point set $(2,P_\phi,c_\phi)$ encoding $\phi$. For readability, we call the two labels ``red'' and ``blue''. The goal is that there is some number of points $n'$, such that there exists a reduced training set with at most $n'-k$ points if and only if there exists an assignment of the variables in $\phi$ fulfilling at least $k$ clauses.

To translate the embedding of $\phi$ into a training set, we show how to encode variables using \emph{variable gadgets}, pass those variables along polylines using \emph{channels}, and how to encode clauses using \emph{clause gadgets}. To ensure that gadgets only interact with each other as intended, each gadget is designed to classify all points outside of some bounded area to be blue.

A variable gadget is given by the labelled point set shown in \Cref{fig:variable_raw}. In each half (upper and lower) of this point set, there are only two strict subsets (\Cref{subfig:variable_raw_all}) which lead to the same classification: The outer points encoding true (\Cref{subfig:variable_raw_true}), and the inner points encoding false (\Cref{subfig:variable_raw_false}). Note that both of these subsets have the same number of points.

\begin{figure}
    \centering
    \begin{subfigure}{1\columnwidth}
    \includegraphics[width=\textwidth,keepaspectratio]{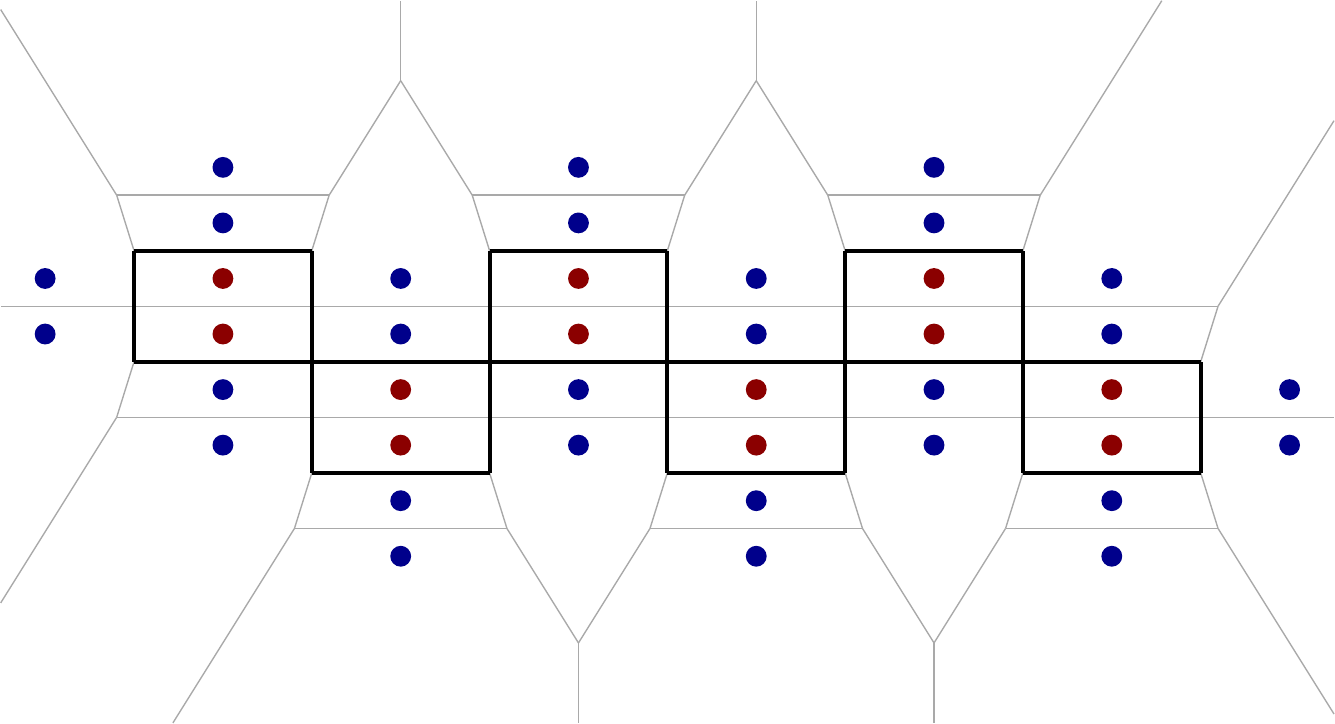}
    \caption{The complete training set.}\label{subfig:variable_raw_all}
    \end{subfigure}
    \begin{subfigure}{0.49\columnwidth}
    \includegraphics[width=\textwidth,keepaspectratio,page=2]{figs/variable_raw.pdf}
    \caption{The ``true'' subset.}\label{subfig:variable_raw_true}
    \end{subfigure}
    \hfill
    \begin{subfigure}{0.49\columnwidth}
    \includegraphics[width=\textwidth,keepaspectratio,page=3]{figs/variable_raw.pdf}
    \caption{The ``false'' subset.}\label{subfig:variable_raw_false}
    \end{subfigure}
    \caption{The variable gadget.}
    \label{fig:variable_raw}
\end{figure}

The 6 uppermost and 6 lowermost points in \Cref{subfig:variable_raw_all} are \emph{shielding points} only used to ensure that every point outside of the gadget is classified blue.
%Pairs of shielding points in the same column can be removed to allow the gadget to interface with a channel, which is used to pass the truth value encoded by the variable gadget to a clause gadget.
Channels are now attached on the top and bottom of the variable gadget to connect the variable gadget to clause gadgets above or below.
If a channel connects to a clause gadget in which the variable occurs positively, the channel is attached to a column without shielding points. Otherwise, the variable occurs negatively, the channel is attached to a column with shielding points. Channels are never attached to the outermost columns, and channels on the same side of the variable gadget always leave at least two columns in between each other unused. Note that a variable gadget can be extended horizontally by adding more of the repeating point pattern, to allow for an arbitrary number of channels.
A variable gadget with two attached channels is shown in \Cref{fig:variable_with_attachments}.

\begin{figure}
    \centering
    \begin{subfigure}{1\columnwidth}
    \includegraphics[width=\textwidth,keepaspectratio]{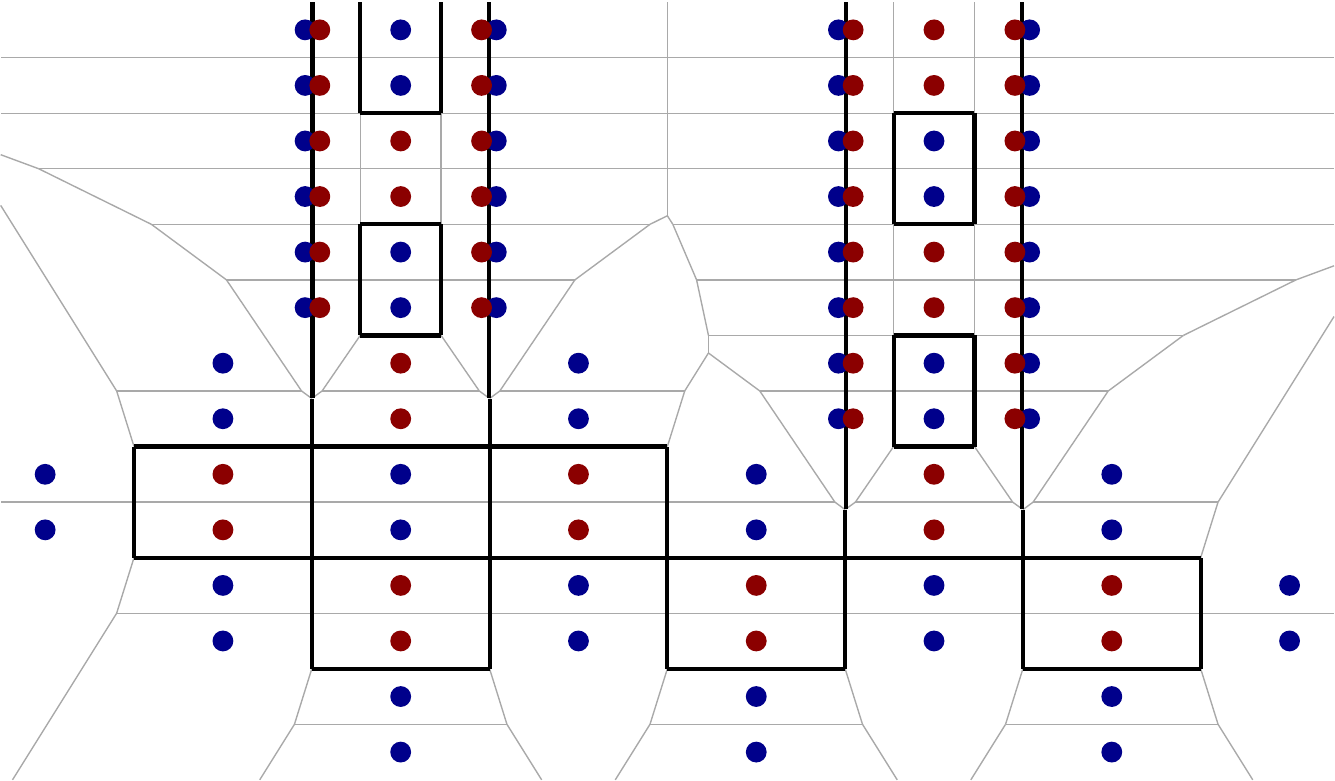}
    \caption{The complete training set.}\label{subfig:variable_with_attachments_all}
    \end{subfigure}
    \begin{subfigure}{0.49\columnwidth}
    \includegraphics[width=\textwidth,keepaspectratio,page=2]{figs/variable_with_attachments.pdf}
    \caption{The ``true'' subset.}\label{subfig:variable_with_attachments_true}
    \end{subfigure}
    \hfill
    \begin{subfigure}{0.49\columnwidth}
    \includegraphics[width=\textwidth,keepaspectratio,page=3]{figs/variable_with_attachments.pdf}
    \caption{The ``false'' subset.}\label{subfig:variable_with_attachments_false}
    \end{subfigure}
    \caption{A variable gadget with two channels attached on the top. The left channel leads to a clause in which the variable occurs positively, and the right channel leads to a clause in which the variable occurs negatively.}
    \label{fig:variable_with_attachments}
\end{figure}

The truth value carried by a channel will be interpreted by a clause gadget by checking for the presence of a certain point $p$ in an arrangement of points introduced at the end of a channel. This arrangement is shown in \Cref{fig:channel_cap}.

\begin{figure}
    \centering
    \begin{subfigure}{0.32\columnwidth}
    \includegraphics[width=\textwidth,keepaspectratio]{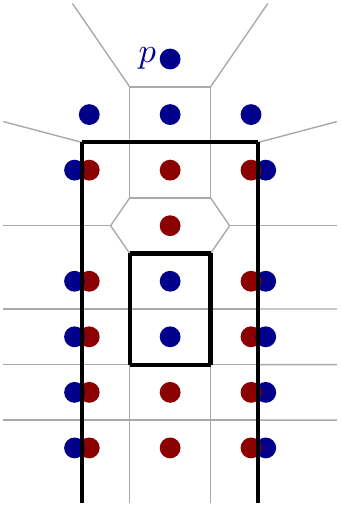}
    \caption{All points.}\label{subfig:channel_cap_all}
    \end{subfigure}
    \begin{subfigure}{0.32\columnwidth}
    \hfill
    \includegraphics[width=\textwidth,keepaspectratio,page=2]{figs/channel_cap.pdf}
    \caption{``true'' subset.}\label{subfig:channel_cap_true}
    \end{subfigure}
    \hfill
    \begin{subfigure}{0.32\columnwidth}
    \includegraphics[width=\textwidth,keepaspectratio,page=3]{figs/channel_cap.pdf}
    \caption{``false'' subset.}\label{subfig:channel_cap_false}
    \end{subfigure}
    \caption{The end of a channel. Point $p$ is present if the value carried by the channel is true.}
    \label{fig:channel_cap}
\end{figure}

We wish to make a short remark on the distances between points in our training set. We say the vertical distance between two neighboring points in a variable gadget (and channel) is~$1$. The horizontal distance between points in a variable gadget is chosen to be $3.2$. This is a very deliberate choice, since on the one hand, a too small horizontal distance would make the point $p$ obsolete in \Cref{subfig:channel_cap_true}. On the other hand, a too large horizontal distance would make some of the pairs of red and blue points that are used to generate the left and right vertical boundary of the channel obsolete. It turns out that $3.2$ is a value that avoids both of these issues simultaneously.

To allow us to connect all variable and clause gadgets, channels need some more flexibility. \Cref{fig:channel_bends} shows how \emph{bends} of some fixed small angle can be achieved in channels. This angle is an irrational constant (close to~$10^\circ$) chosen such that its sine and cosine are rational numbers. This ensures that all coordinates of points are rational. We can also \emph{stretch} channels longitudinally (in the direction of their repeating pattern), by simply increasing the distance between a row with blue center point and a row with red center point. With the capability of creating bends and stretching longitudinally, we can make the endpoint of a channel lie at the locations needed to attach to the clause gadgets.

\begin{figure}
    \centering
    \begin{subfigure}{0.32\columnwidth}
    \includegraphics[width=\textwidth,keepaspectratio]{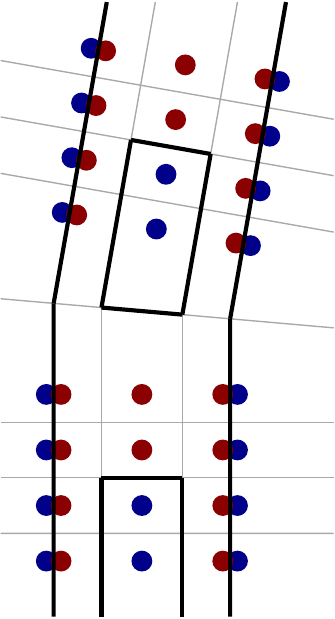}
    \caption{All points.}\label{subfig:channel_bends_all}
    \end{subfigure}
    \hfill
    \begin{subfigure}{0.32\columnwidth}
    \includegraphics[width=\textwidth,keepaspectratio,page=2]{figs/channel_bends.pdf}
    \caption{``true'' subset.}\label{subfig:channel_bends_true}
    \end{subfigure}
    \hfill
    \begin{subfigure}{0.32\columnwidth}
    \includegraphics[width=\textwidth,keepaspectratio,page=3]{figs/channel_bends.pdf}
    \caption{``false'' subset.}\label{subfig:channel_bends_false}
    \end{subfigure}
    \caption{A bend of $\approx 10^\circ$ in a channel.}
    \label{fig:channel_bends}
\end{figure}

A clause gadget is shown in \Cref{fig:clause}. It contains two special points, marked $a_1$ and $a_2$. These points are the locations of the endpoints of the two channels carrying the values of the involved variables to the clause gadget. Note that to be able to fit these channels without disturbing the function of the clause gadget, $a_1$ and $a_2$ need to have small vertical distance (say, $0.5$) and large horizontal distance (say, $5$). The clause gadget can also be drawn with these distances, but for legibility, the horizontal and vertical distances have been equalized in \Cref{fig:clause}. We include a figure showing the clause gadget with correct distances in \Cref{app:clausegadget}.

If the value carried by the first (second) channel is true, the point $a_1$ ($a_2$) is part of the reduced training set, and otherwise it is not. The clause gadget is built in such a way that it needs $5$ points if neither of $a_1$ and $a_2$ is present (\Cref{subfig:clause_noActiv}), and $4$ points otherwise (\Cref{subfig:clause_activA,subfig:clause_activB,subfig:clause_activBoth}). Thus, we can save one point in the clause gadget if and only if the corresponding clause is fulfilled by the variable assignment corresponding to the points present in the channels.

\begin{figure}
    \centering
    \begin{subfigure}{1\columnwidth}
    \centering
    \includegraphics[width=0.8\textwidth,keepaspectratio]{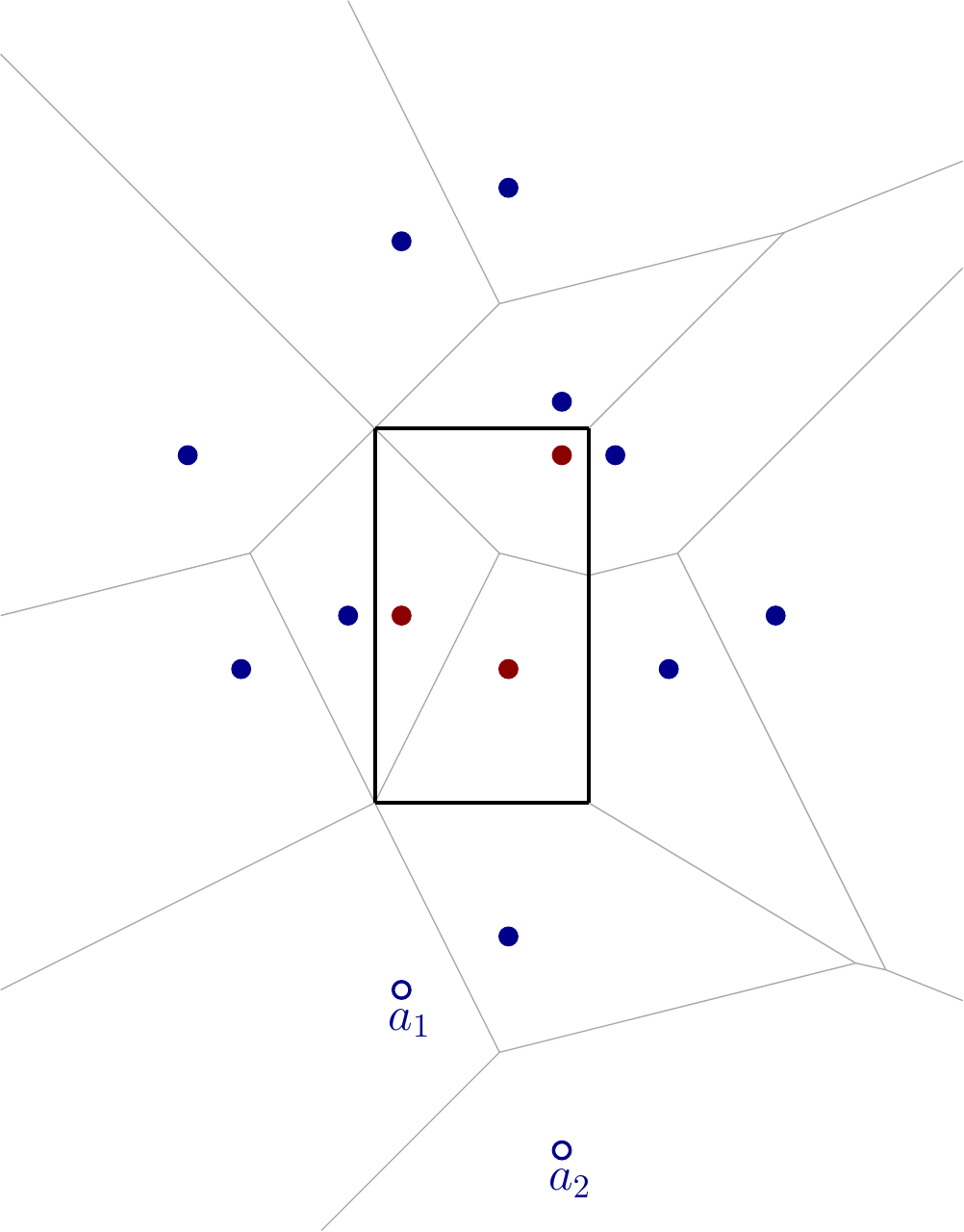}
    \caption{The complete training set.}\label{subfig:clause_all}
    \end{subfigure}
    \begin{subfigure}{0.49\columnwidth}
    \includegraphics[width=\textwidth,keepaspectratio,page=2]{figs/clause.pdf}
    \caption{Both variables false.}\label{subfig:clause_noActiv}
    \end{subfigure}
    \hfill
    \begin{subfigure}{0.49\columnwidth}
    \includegraphics[width=\textwidth,keepaspectratio,page=4]{figs/clause.pdf}
    \caption{First variable true.}\label{subfig:clause_activA}
    \end{subfigure}
    \begin{subfigure}{0.49\columnwidth}
    \includegraphics[width=\textwidth,keepaspectratio,page=3]{figs/clause.pdf}
    \caption{Second variable true.}\label{subfig:clause_activB}
    \end{subfigure}
    \hfill
    \begin{subfigure}{0.49\columnwidth}
    \includegraphics[width=\textwidth,keepaspectratio,page=5]{figs/clause.pdf}
    \caption{Both variables true.}\label{subfig:clause_activBoth}
    \end{subfigure}
    \caption{The clause gadget.}
    \label{fig:clause}
\end{figure}

We have now introduced all necessary gadgets and techniques for the proof of \Cref{thm:NP}

\begin{proof}
For any fixed $d$, we can verify that a given subset of points is a reduced training set by computing the Voronoi diagrams of the training set and the subset and comparing the classifications. As the Voronoi diagram of $n$ points in $\R^d$ can be computed in polynomial time~\cite{chazelle1993convexhull}, this proves \NP-containment.

\newpage Towards proving \NP-hardness, we reduce from V\nobreakdash-cycle max2SAT, as in \Cref{def:vcycle2SAT}. Given an instance $(\phi,k)$ of V-cycle max2SAT, we can find a $2$-book embedding of $G'_\phi$ in polynomial time. This is translated to an embedding of the bipartite variable-clause graph of $\phi$ as in \Cref{fig:graphembedding}. Every box corresponding to a variable is replaced by a variable gadget, every box corresponding to a clause is replaced by a clause gadget, and the edges are replaced by channels (possibly with up to $2$ bends and some number of stretchings). This point set can be constructed in polynomial time, since it contains a polynomial number of points, and all points can be placed on rational coordinates requiring at most polynomially many bits to describe.

Let $n_1$ be the number of points in the reduced training set of all variable gadgets and channels corresponding to some fixed assignment of truth values to variables. Let $n_2$ be $5\cdot b$ (recall that $b$ is the number of clauses in $\phi$). We will now prove that there exists a minimum cardinality reduced training set of size at most $n_1+n_2-k$ if and only if there exists an assignment of variables fulfilling at least $k$ clauses of $\phi$.

The ``if'' direction is trivial: If such an assignment of variables exists, it can clearly be translated into a reduced training set of the correct size, since each fulfilled clause gadget only requires four points.

For the ``only if'' direction, we argue that any reduced training set can be turned into a reduced training set corresponding to a variable assignment, without increasing the number of points. %\\
We first consider the channels. If any channel (including its end) is using any reduced subset other than the ``true'' or ``false'' subset, it must be using additional points. We change the subset to be the subset matching the truth value carried by the side of the variable gadget the clause is attached to. This may cost us one additional point in the clause gadget, but will also save at least one point in the channel. %\\
Next, we fix the variables. If any variable gadget is using any reduced subset other than the ``true'' or ``false'' subset, it is using all points on at least one of the two (upper and lower) halves, and in the other half it must be using either (a) all points, or (b) only the ``false'' subset. Let $m$ be the number of columns in the gadget. In case (a), we can switch the gadget to using the ``false'' subset in both halves, and the attached channels and clauses are switched accordingly. This may cost one point per connected clause gadget, but saves at least $2m$ points in the variable gadget, which is strictly more since not all columns can be occupied by channels on each side. In case (b), we switch the gadget to using the ``true'' subset in both halves. This may cost one point per clause gadget connected to the side that was previously false, but saves at least $m$ points in the variable gadget, which is again strictly more.
We thus conclude that our training set has at least one minimum cardinality reduced training set which corresponds to a variable assignment, proving the correctness of our reduction.
\end{proof}

\small{\paragraph{Acknowledgments.} We thank David Eppstein for his great talk at SOSA'22 that inspired this work, and Bernd Gärtner for his valuable advice.}
%---------------------------- Bibliography -------------------------------
%\newpage
% Please add the contents of the .bbl file that you generate,  or add bibitem entries manually if you like.
% The entries should be in alphabetical order
\small
\bibliographystyle{abbrv}
\bibliography{literature}

\clearpage

\appendix
\onecolumn
\section{Clause Gadget}\label{app:clausegadget}
Below, in \Cref{fig:true_clause}, the clause gadget is shown to scale, with the distance between the two attachment point for channels $a_1$ and $a_2$ having horizontal distance of $5$ and vertical distance of $0.5$. The vertical stretching of the gadget is necessary to ensure that the classification remains correct in the case where both variables are true (\Cref{subfig:true_clause_both}). If the gadget would be less stretched, a Voronoi wall between the red point and $a_2$ would appear.
\begin{figure}[h]
    \centering
    \begin{subfigure}{0.45\textwidth}
        \centering
        \includegraphics[width=\textwidth,keepaspectratio]{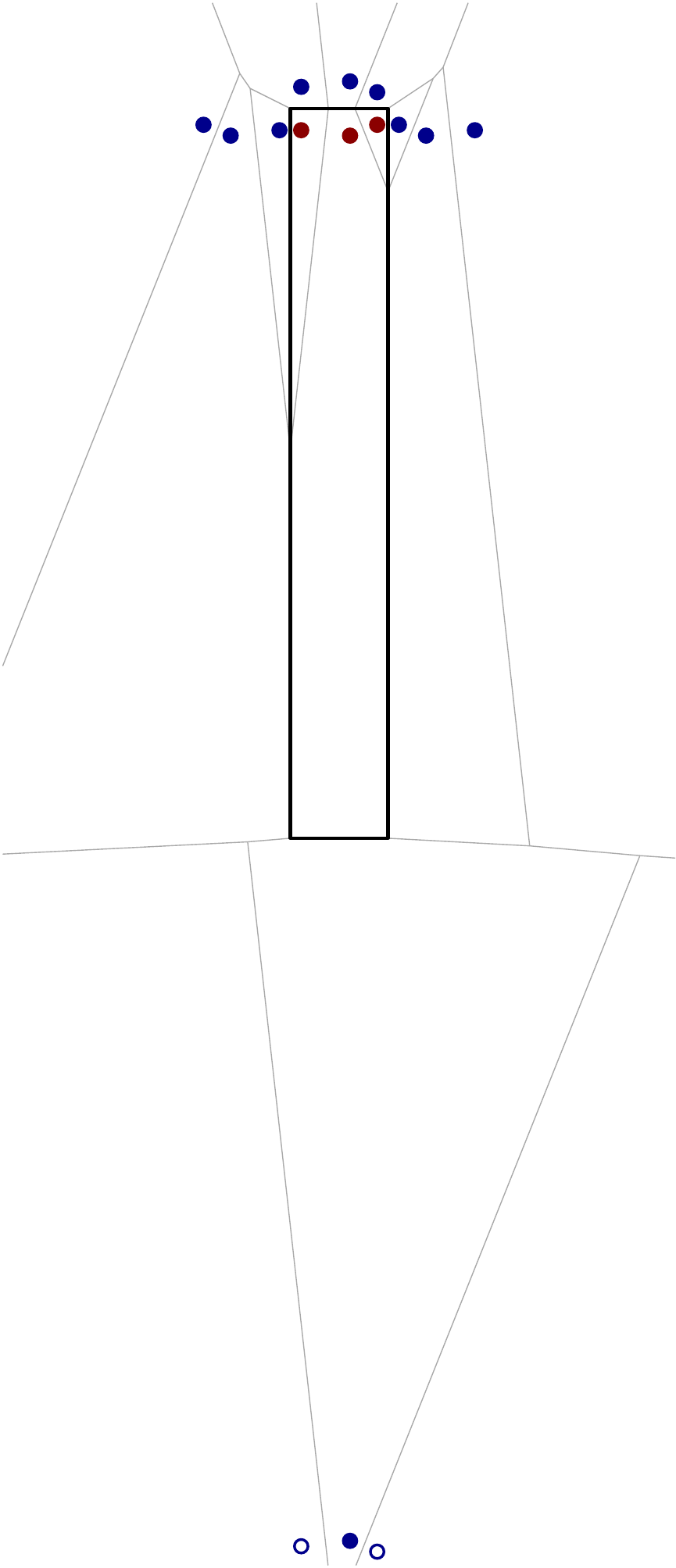}
        \caption{The complete training set.}
        \label{subfig:true_clause_all}
    \end{subfigure}
    \hfill
    \begin{subfigure}{0.45\textwidth}
        \centering
        \includegraphics[width=\textwidth,keepaspectratio,page=5]{figs/clause_true.pdf}
        \caption{Both variables are true.}
        \label{subfig:true_clause_both}
    \end{subfigure}  
    \caption{The clause gadget with the correct distances between the attachment points.}\label{fig:true_clause}
\end{figure}

\end{document}